\documentclass{amsproc}

\newtheorem{theorem}{Theorem}[section]

\theoremstyle{definition}
\newtheorem{definition}[theorem]{Definition}

\theoremstyle{remark}
\newtheorem{remark}[theorem]{Remark}

\theoremstyle{property}
\newtheorem{property}[theorem]{Property}

\numberwithin{equation}{section}

\def\be{\begin{equation}}
\def\ee{\end{equation}}
\def\bea{\begin{eqnarray}}
\def\eea{\end{eqnarray}}
\def\ba{\begin{array}{rcl}}
\def\ea{\end{array}}

\catcode`@=11 
\font\sevensy=cmsy7
\catcode`@=12 
\newbox\novebox
\setbox\novebox=\hbox{\kern+.3em {$V_i$}%
\kern-0.68em\raise1.8ex\hbox{\sevensy {\char'016}}\kern+.3em} 
\def\vint{\copy\novebox}
\usepackage{graphicx}
\usepackage{amssymb,amsmath}
\usepackage{epsfig}

\begin{document}

\title{Abelian link invariants and homology}

\author{Enore Guadagnini}
\address{Dipartimento di Fisica dell'Universit\`a di Pisa e INFN Sezione di Pisa}
\curraddr{Dipartimento di Fisica ``E. Fermi",
Largo B. Pontecorvo 3, 56127, Pisa, Italy}
\email{guada@df.unipi.it}
\author{Francesco Mancarella}
\address{International School for Advanced Studies SISSA, Italy}
\curraddr{SISSA, Via Beirut 2-4, 34151 Trieste, Italy}
\email{mancarel@sissa.it}


\begin{abstract}
We consider the link invariants defined by the quantum Chern-Simons field theory with compact  gauge group $U(1)$ in a closed oriented 3-manifold $M$. The relation of the abelian link  invariants with  the homology group of the complement of the links is discussed.  We prove that, when $M$ is a homology sphere or when  a link ---in a generic manifold $M$--- is homologically trivial, the associated observables coincide with the observables of the sphere $S^3$.  Finally we show that the $U(1)$ Reshetikhin-Turaev surgery invariant of the manifold  $M$  is not a function of  the homology group only, nor a function of the homotopy type of $M$ alone. \end{abstract}

\maketitle

\section{Introduction}
Quantum field theories can be used not only to describe the physics of elementary particles but also to compute topological invariants. In this article we consider the link invariants that are defined by the abelian $U(1)$ Chern-Simons field theory formulated in a closed and oriented 3-manifold $M$, and we show  how  these invariants are related with the homology group of the complement of the links. To this end, we shall introduce a few definitions ---like that of simplicial satellite or of equivalent knot--- which are used to connect the values of the  $U(1)$-charges which are associated with the components of the links with the numbers that classify the homology classes of loops. 

We  demonstrate that the set of the abelian link invariants (or  observables) in any homology sphere coincides with the set of observables in the sphere $S^3$. We also prove that if a link in a generic manifold $M$ is  homologically trivial then its associated observable coincides with an observable computed in  $S^3$. We then consider the $U(1)$ surgery  invariant  of Reshetikhin-Turaev, we show that this invariant: (1) is trivial for homology spheres, (2) is not a function of the homology group of the manifold only  and (3) is not a function of the homotopy type of the manifold only.   

In order to make this article self-contained, we have added a preliminary section containing a description of the new developments in the field theory computations of the link invariants together with a brief description of the Reshetikhin-Turaev surgery rules.

\section{Field theory approach}

The abelian Chern-Simons theory \cite{ASS, ASS2, CRH,P}  is a gauge theory defined in terms of a $U(1)$-connection $A$ in a closed oriented 3-manifold $M$. For each oriented knot $C\subset M$, the corresponding holonomy   is given by the integral $\int_C A\, $ which is invariant under  $U(1)$ gauge transformations acting on $A$. 

In the standard field theory formulation of abelian gauge theories, the (classical fields) configuration space locally coincides with the set of 1-forms modulo exact forms, $A \sim A + d \Lambda$. However, if one assumes \cite{W1, FR} that a complete set of observables is given by the exponential of the holonomies $\{ \exp [  2 \pi i \int_C A ] \}$ which are associated with oriented knots  $C $  in $M$, the invariance group of the observables is actually  larger than the standard gauge group.   In facts,  the observables must be locally defined on the classes of 1-forms modulo forms $\widehat A$ with integer periods, $A\sim A+ \widehat A  \; ,  \; \int_C {\widehat A} = n\in {\mathbb Z}$. This means that the configuration space is defined in terms of the Deligne-Beilinson cohomology classes \cite{D, B, FR}. So, we shall now consider the Deligne-Beilinson (DB) formulation of the abelian $U(1)$ Chern-Simons gauge theory. 

In order to simplify the notation,  the classes belonging to the DB cohomology group of $M$ of degree~1,  $H^1_D(M) $, will be denoted by $A$. Let $H^3_D(M) $ be the space of the DB classes of degree~3. 
The pairing of the DB cohomology groups, which is called the *-product, defines a natural mapping \cite{PR}
\be
H_D^1(M) \otimes H_D^1(M) \longrightarrow H^3_D(M) \; . 
\label{1}
\ee
The *-product of $A$ with $A$ just corresponds to the abelian Chern-Simons lagrangian \cite{FR,EF} 
\be 
A *A \longrightarrow A\, \wedge d A \; . 
\label{2}
\ee
Precisely like the integral of any element of $H^3_D(M)$, the Chern-Simons action 
\be 
S= \int_M A* A  \longrightarrow \int_M A\wedge d A 
\label{3}
\ee
is defined modulo integers; consequently, the path-integral phase factor
\be
\exp \bigl \{ 2 \pi i k S \bigr \} = \exp \left \{ 2 \pi i k  \int_M A * A \right \}
\label{4}
\ee
is well defined when the coupling constant $k$  takes integer values 
\be 
k \in {\mathbb Z} \quad , \quad k \not= 0 \; . 
\label{5}
\ee
A modification of the orientation of $M$ is equivalent to the replacement $k \rightarrow - k$.  Let us  consider a framed, oriented and coloured link $L\subset M$ with $N$ components $\{ C_1 , C_2 , ... , C_N \} $.  The colour of each component $C_j$, with $j=1,2,...,N$, is represented by an integer charge $q_j \in {\mathbb Z}$.   The classical expression $W(L)$ of the Wilson line is given by 
\be
W(L) =  \prod_{j=1}^N \exp \Bigl \{ 2 i \pi q_j \int_{C_j} A  \Bigr \}
= \exp \Bigl \{ 2 i \pi \sum_j q_j \int_{C_j}  A \Bigr \} \; .
\label{6}
\ee
Each link component which has colour $q=0$ can be eliminated, and a modification of the orientation of a link component $C$ is equivalent to a change of the sign of the 
associated charge $q$.  The observables of the Chern-Simons gauge theory in $M$ are given by the expectation values 
\begin{equation} 
\langle W(L) \rangle \Bigr |_{M} = \frac{\int_M DA \, e^{2 \pi i k S } \, W(L) }{\int_M DA \, e^{2 \pi i k S }} \; , 
\label{7}
\ee
where the path integral  should be defined on the DB classes  which belong to $H^1_D(M) $.  More precisely, the structure of the functional space admits a natural description in terms of the homology groups of $M$, as indicated by the following exact sequence \cite{S1, S2}
\be
0\buildrel \over \longrightarrow {\Omega ^1\left( {M } \right)}
\mathord{\left/ {\vphantom {{\Omega ^1\left( {M } \right)} {\Omega _{\mathbb Z}^1 \left( {M} \right)}}} \right. \kern-\nulldelimiterspace} {\Omega
_{\mathbb Z}^1 \left( {M } \right)}\buildrel \over \longrightarrow H_D^1
\left( {M} \right)\buildrel \over \longrightarrow H^2\left( {M} \right)\buildrel \over \longrightarrow 0 \; , 
\label{8}
\ee
where $\Omega^1( M )$ is the space of  $1$-forms on $M$, $\Omega _{\mathbb Z}^1 ( M )$ is the space of closed $1$-forms with integer periods on $M$ and $H^{p}( M )$ is the $( p)^{th}$ integral cohomology group of $M$. Thus, $H_D^1( M)$ can be understood as an affine bundle over $H^2(M)$, whose fibres have a typical underlying (infinite dimensional) vector space structure given by $\Omega^1(M) / \Omega^1_{\mathbb Z}(M)$. 

The framing of the link components is used to fix the ambiguities, which appear in the computation of the expectation values (\ref{7})  of the composite Wilson line operators, in such a way to maintain the ambient isotopy invariance of the expectation values \cite{AM,L, EF}. 

Assuming that expression (\ref{7}) is well defined, one can prove \cite{EF} the most important properties of the expectation values:  (i) the colour periodicity, (ii) the ambient isopoty invariance and (iii) the validity of the satellite relations. We shall briefly discuss these subjects in section~2.3. When expression (\ref{7}) is well defined, the computation of the observables  provides the solution of the Chern-Simons field theory in the manifold $M$. 

\subsection{Fundamental link invariants}

When the 3-manifold $M$ coincides with the 3-sphere $S^3$, one can compute the expectation values (\ref{7})  by means of (at least) two methods: standard perturbation theory or a nonperturbative path integral computation. Both methods give the same answer. 

\subsection*{First method.} Since the topological properties of links in ${\mathbb R}^3$ and in $S^3$ coincide, let us consider the abelian Chern-Simons theory formulated in  ${\mathbb R}^3$. In this case, the Deligne-Beilinson approach is equal to the standard perturbative formulation of the abelian gauge theories. The direct computation of the observables (\ref{7}) by means of standard perturbation theory \cite{L} gives
\be
\langle W(L) \rangle \Bigr |_{S^3} = \langle W(L) \rangle \Bigr |_{{\mathbb R}^3}   =   \exp \Bigl \{ -( 2 i \pi / 4k)  \sum_{ij} q_i  {\mathbb L}_{ij} q_j  \Bigr  \} \,  ,
\label{9}
\ee
where the off-diagonal elements of the linking matrix $ {\mathbb L}_{ij}$, which is associated with the link $L$, are given by the  linking numbers between the different link components 
\be
 {\mathbb L}_{ij} =   {\ell} k ( C_i , C_j )  =   {\ell} k ( C_j , C_i )\quad , \quad {\rm for~}  i \not= j ; 
\label{10}
\ee
whereas the diagonal elements of  the matrix $ {\mathbb L}_{ij}$  correspond to the linking numbers of the link components $\{ C_j \}$  with their framings $\{ C_{j{\rm f}} \}$
 \be
  {\mathbb L}_{jj} = {\ell} k ( C_j , C_{j{\rm f}} ) = {\ell} k ( C_{j{\rm f}}, C_j )\; .  
   \label{11}
\ee

\subsection*{Second method.} 
Sequence (\ref{8}) implies  that $H_D^1( S^3) \simeq \Omega^1(S^3) / \Omega^1_{\mathbb Z}(S^3)$ because $H^2 (S^3)$ is trivial. By using the property of translation invariance of the functional measure,  which can also be expressed in the form of a Cameron-Martin like formula \cite{CM}, one can introduce \cite{EF} a change of variables in the numerator of (\ref{7})  in such a way to factorize out the  value of the partition function, which cancels with the denominator.   
As a result, one can produce an explicit nonperturbative path-integral computation of the observables (\ref{7}) and one finds  \cite{EF}
\be
\langle W(L) \rangle \Bigr |_{S^3}   =   \exp \Bigl \{ -( 2 i \pi / 4k)  \sum_{ij} q_i  {\mathbb L}_{ij} q_j  \Bigr  \} \,  , 
\label{12}
\ee
which coincides with expression (\ref{9}). 

The observables (\ref{12}), which are  ambient isotopy invariants,  are called the abelian link invariants. They represent the fundamental invariants because, as we shall see,  the value of any other topological invariant of the abelian Chern-Simons theory in a generic 3-manifold $M$ can be derived from expression (\ref{12}).   

\subsection{Observables computation}

When the Chern-Simons field theory is defined in a nontrivial manifold $M$, the explicit computation of the observables by means of the standard field theory formulation of gauge theories presents some technical difficulties, which are related to the gauge-fixing procedure and the definition of the fields propagator. For example, when $M= S^1 \times S^2$, the  Feynman propagator for the $A$ field does not exist because of the presence of a physical zero mode; in facts, among the field configurations, a globally defined  1-form $A_0$ exists such that $d A_0 = 0 $ but $A_0$ is not the gauge transformed of something else.  One can presumably overcome  these technical difficulties, and one can imagine of computing the observables by means of perturbation theory.  But, as a matter of facts, an explicit path-integral computation of the link observables (\ref{7})  by means of the standard gauge theory perturbative methods has never been produced when the 3-manifold is not equal to $ {\mathbb R}^3$. 

For a nontrivial 3-manifold $M$, the expectation values $\langle W(L) \rangle \bigr |_{M}$   can be   really computed ---for certain manifolds--- by using two methods: (i) a nonperturbative path-integral formalism based on the Deligne-Beilinson cohomology, (ii) the operator surgery method. In all the cases considered so far, these two methods give exactly the same answer. 

\subsection*{Nonperturbative path-integral computation.}
Let us consider a class  of torsion-free manifolds of the type $S^1 \times \Sigma$, where $\Sigma $ denotes the 2-sphere $S^2$ or a closed Riemann surface of  genus $g \geq 1$. In this case, the first homology group $H_1(M)$ is not trivial and is given by the product of free abelian group factors;  standard perturbation theory cannot be used since  the Feynman   propagator for the $A$ field does not exist in $S^1 \times \Sigma$. But one can use the nonperturbative method developed in \cite{EF}, in which the introduction of a gauge fixing and of the Feynman propagator is not necessary. 
The structure of the bundle $H_D^1( M)$, which is determined by the sequence (\ref{8}),  and of the resulting path-integral have been described in \cite{EF}. One finds: 

\begin{enumerate}

\item  when $L$ is not homologically trivial (mod $2k$) in $S^1 \times \Sigma$,  
\be
 \langle W(L) \rangle \Bigr |_{S^1\times \Sigma}   =  0  \; ;  
 \label{13}
 \ee
 
\item when $L$ is homologically trivial (mod $2k$)  in $S^1 \times \Sigma$, 
\be
\langle W(L) \rangle \Bigr |_{S^1 \times \Sigma} =   \exp \Bigl \{ -( 2 i \pi / 4k)  \sum_{ij} q_i  {\mathbb L}_{ij} q_j  \Bigr  \} \,  , 
\label{14}
\ee 

\end{enumerate}

\noindent which formally coincides with expression (\ref{12}). Note that, when $L$ is homologically trivial (mod $2k$), expression (\ref{14}) is well defined \cite{EF}. 

By using nonperturbative path-integral arguments, the results  shown in equations (\ref{13}) and (\ref{14}) have been generalized  by Thuillier \cite{F} to the case in which the 3-manifold is $M= RP^3$. This example is interesting because $H_1 ( RP^3) $ is not freely generated (in fact, $H_1 (RP^3) = {\mathbb Z}_2$) and then $RP^3$ has nontrivial torsion. 

 \subsection* {Operator surgery method.} By means of the quantum groups modular algebra, one can construct link invariants of ambient isotopy; in order to compute these invariants in a nontrivial manifold $M$,  Reshetikhin and Turaev  have introduced appropriate surgery rules  \cite{RT}. These rules  ---that  have been also developed by Kohno \cite {K}, by Lickorish \cite{LIC} and by Morton and Strickland \cite{MOST} in the mathematical setting---  have been adapted to the physical context in \cite{W1, OE, OEP, L}. We shall now recall the main features of the operator surgery method, which can be used to compute  the abelian link invariants in a generic manifold $M$.

Every closed orientable connected 3-manifold $M$ can be obtained by Dehn  surgery on $S^3$ and admits a surgery presentation  \cite{ROL} which is described by a framed surgery link ${\mathcal L}\subset S^3$. A  so-called surgery coefficient $a_i$  is associated with each component ${\mathcal L}_i $ of ${\mathcal L}\, $; when $a_i$ is an integer, we will put $a_i = \ell k ({\mathcal L}_i , {\mathcal L}_{i {\rm f}} )$.   For each manifold $M$, the corresponding surgery link ${\mathcal L}$ is not unique; all the possible surgery links which describe ---up to orientation-preserving homeomorphisms--- the same manifold are related by Kirby moves \cite{ROL}.    Any oriented coloured framed link $L \subset M$ can be described by a link $L^{\, \prime } = L \cup {\mathcal L} $ in $S^3$ in which:

\begin{itemize}

\item the surgery link  ${\mathcal L}$ describes the surgery instruction corresponding to a presentation of $M$ in terms of Dehn surgery on $S^3$;

\item the  link $L$, which belongs to the complement of ${\mathcal L}$ in $S^3$, describes how $L$ is actually placed in $M$.

\end{itemize}

\noindent According to the rules  \cite{L} of the operator surgery method,  the expectation value of the Wilson line operator $W(L)$ in $M$ can be written as a ratio
\be
\langle W ( L) \rangle \Bigr |_{M} =   \langle W ( L) \, W({\mathcal L}) \rangle  \Bigr |_{S^3}  \; {\Big  / } \;   \langle  W({\mathcal L}) \rangle \Bigr |_{S^3} \, ,
\label{15}
\ee
where to each component of the surgery link ${\mathcal L}$ is associated a  particular colour state $\psi_0$. Expression (\ref{12}) implies that, for fixed integer $k$, the colour space  of each link component coincides with  space of residue classes of integers mod~$2k$ (see also section~2.3). Thus the colour space has a canonical ring  ${\mathcal R} $ structure; let $\chi_j$ denote the residue class associated with the integer $j$. Then,  when the gauge group is $U(1)$, the colour state $\psi_0 \in {\mathcal R }$ is given by
\be
\psi_0 = \sum_{j=0}^{2k -1}\,  \chi_j \, .
\label{16}
\ee
This simply means that, in the computation of the observables ({\ref{15}), one must sum over the values $q=0,1,2,..., 2k-1$ of the colours which are associated with the components of the surgery link (see for instance equation (\ref{35})).  
The proof  that the surgery rules (\ref{15}) and (\ref{16}) are well defined and consistent ---when the denominator of expression (\ref{15}) is not vanishing--- is nontrivial and essentially consists in proving that expression  (\ref{15}) is invariant under Kirby moves \cite{KI, RT, L}. 

\begin{remark}
The existence of surgery rules for the computation of the observables in quantum field theory is quite remarkable. Let us summarize the reasons for the existence of surgery rules in the Chern-Simons theory. 
Any 3-manifold $M$ can be obtained \cite{ROL} by removing and gluing back ---after the introduction of appropriate homeomorphisms  on their boundaries--- solid tori embedded in $S^3$.  The crucial point now is that the set of all possible surgeries is generated by two elementary operations which in facts correspond to twist homeomorphisms \cite{ROL}.  The action of  these two twist homeomorphism generators on the observables can be found by analysing expression (\ref{12}).  This means that  the solution of the Chern-Simons field theory in $S^3$ determines \cite{L} the representation of the surgery on the set of observables. As a result, one can then connect the values of the observables in any nontrivial 3-manifold $M$ with the values of the observables in $S^3$. For this reason, the solution of the  topological Chern-Simons  field theory in $S^3$ actually fixes the solution of the same theory in any closed oriented 3-manifold $M$. 
\end{remark}

\subsection{Main properties}

We conclude this section by recalling a few properties of the observables that will be useful for the following discussion.  Since the linking numbers take integer values, expression (\ref{12}) is invariant under the replacement $q_j \rightarrow q_j + 2k$, where $q_j$ denotes the colour of a generic link component.  Thus, for fixed $k$,  the colour space of each link component can be identified with ${\mathbb Z}_{2k}$, which coincides with the space of the residue classes of integers mod~$2k$. This property also holds \cite{EF} for the observables in a generic manifold $M$. 

 At the classical level, one link component $C$ with colour $q > 1$ can be interpreted as the $q$-fold covering of $C$.  At the quantum level, one needs to specify this correspondence a bit more precisely because of possible ambiguities in the computation of the expectation values of the composite Wilson lines operators. As we have already mentioned, all these ambiguities are removed by means of the framing procedure. 

\subsection* {Satellites.}  A general discussion of the satellite properties  of the observables of the Chern-Simons theory can be found in Ref.\cite{W1, RT, LIC,  MOST, L}. Here we shall concentrate on the aspects which are relevant for the following exposition.  
 
 Let $C_{\rm f}$ be the framing of the oriented link component $C\subset M$ which has colour  $q $ with  $|  q | > 1 $. One can imagine that $C$ and $C_{\rm f}$ define the boundary of a band ${\mathcal B}\subset M$;  then, one can \cite{L, EF} simply replace $C$  with $|q|$ parallel components $\{ \widetilde C_1,..., \widetilde C_{|q|}\} $ on ${\mathcal B}$ where each component has colour $q^\prime = 1$ (in order to agree with the sign of $q$, one possibly needs to modify the orientations of the link components). The framings 
$\{ \widetilde C_{1{\rm f}},..., \widetilde C_{{|q|}{\rm f}} \} $ of the components $\{ \widetilde C_1,..., \widetilde C_{|q|}\} $ also belong to the band ${\mathcal B}$. One can easily verify that the observables (\ref{12}) are invariant under this substitution. To sum up, as far as the abelian link invariants are concerned, each link component $C$ with colour $| q| > 1 $ can always be interpreted as (and can be substituted with) the union of $|q|$ parallel copies of $C$ with unitary colours. 

\begin{definition} For any coloured, oriented and framed link $L\subset M$, one can  introduce a new link $\widetilde L\subset M$ which is a satellite of $L$ and which is obtained from $L$ by replacing each link component of colour $q$ with $|q|$ parallel copies of the same component, each copy with unitary colour. We call $\widetilde L $ the {\it simplicial satellite} of $L$.  
\end{definition}

The observables associated with any link $L$ and the observables associated with its simplicial satellite $\widetilde L$ are totally equivalent.  In other words,   the observables of the abelian Chern-Simons theory in a generic manifold $M$ satisfy \cite{EF} the relation  
 \be
 \langle W(L) \rangle \Bigr |_{M} = \langle W(\widetilde L) \rangle \Bigr |_{M}  \; , \; \forall L\subset M  \; . 
 \label{17}
 \ee 
 The introduction of the simplicial satellites is useful because, in this way,   we can possibly do without the concept of colour space, which has not a topological nature, and  we can interpret the abelian link invariants entirely in terms of homology groups. This issue will be discussed in the next section.

\section{Homology and link complements}

In this section we show that, for any link $L\subset S^3 $ with simplicial satellite ${\widetilde L}$, the abelian link invariant $\langle W( L) \rangle \bigr |_{S^3} $ is completely determined by the homology group  $H_1( S^3 - {\widetilde L} )$ of the complement of the link $\widetilde L$ in $S^3$. We also prove that
\begin{itemize}
\item the sets of the abelian Chern-Simons observables  in each  homology  3-sphere and in $S^3$ coincide;    
\item if the simplicial satellite of a link in a generic 3-manifold is homologically trivial, the associated observable coincides with an observable in  $S^3$.   
\end{itemize}

\subsection {Link complements} Let us  firstly recall  that the homology group $H_1(X)$ of a manifold $X$ can be interpreted as the abelianization of the fundamental  group $\pi_1 (X)$ because, given a presentation of $\pi_1(X)$ in terms of generators $\{  \gamma_1, \gamma_2,...\}$ and a set of relations between them, by adding the new constraints $ [\gamma_a , \gamma_b] =0 $ for all $a$ and $b$, one obtains a presentation of $H_1(X)$. Thus, let $C_1$ be an oriented knot in $S^3$; the homology group  of its complement $X= S^3 - C_1$ is freely generated, $H_1(S^3 - C_1)= {\mathbb Z}$, and one can represent the generator $g_1$ by means of a small oriented circle $C_{g_1}$  in $S^3$  linked with $C_1$ so that  $ {\ell} k ( C_{g_1}, C_1 ) =1 $.  Consider now a second oriented knot $C_2 \subset S^3 - C_1$, the class $[C_2]$ of $C_2$ in  $H_1(S^3 - C_1)$ is just determined by the linking number of $C_1 $ and $C_2$. Indeed, by using additive notations, one has 
\be 
[C_2 ] = n \, g_1 \quad \Longleftrightarrow \quad {\ell} k ( C_2, C_1 )=n \; . 
\label{20}
\ee
Moreover, if $C_{2{\rm f}} $ is a framing for $C_2$, one finds 
\be
H_1(S^3 - C_1) \ni [C_{2{\rm f}} ] = [C_2 ]  \quad , \quad {\ell} k ( C_2, C_1 )=  {\ell} k (C_{2{\rm f}} ,  C_1 )\; . 
\label{21}
\ee

Let us now consider  a framed, oriented and coloured link $L$ with simplicial satellite ${\widetilde L}\subset S^3$, the associated abelian link invariant is given by 
\be 
\langle W( L) \rangle \bigr |_{S^3} = \langle W({\widetilde L}) \rangle \bigr |_{S^3} = 
\exp \Bigl \{ -( 2 i \pi / 4k)  \sum_{ij}  {\widetilde {\mathbb L}}_{ij}  \Bigr  \} \,  . 
\label{22}
\ee
$ {\widetilde {\mathbb L}}_{ij} $ denotes the linking matrix of $\widetilde L$ which can be written as 
\be 
 {\widetilde {\mathbb L}}_{ij} =  {\ell} k ( \widetilde C_{i {\rm f}} , \widetilde C_j  )\; , 
 \label{23}
 \ee
 where $ \widetilde C_j $ represents the $j$-th component of ${\widetilde L}$ and $ \widetilde C_{i {\rm f}} $ is the framing of the component  $ \widetilde C_i \subset {\widetilde L}$. If $g_j $ denotes  the $j$-th generator of $H_1(S^3 - {\widetilde L} ) $ which is associated with the component $\widetilde C_j$, the class  $ [ \widetilde C_{i {\rm f}} ] \in H_1(S^3 - {\widetilde L} )$ of the component $\widetilde C_{i {\rm f}}$  can be written as 
 \be 
  [ \widetilde C_{i {\rm f}} ] = \sum_j {\ell} k ( \widetilde C_{i {\rm f}} , \widetilde C_j  )\,  g_j \; . 
  \label{24}
 \ee
So, the class  $[{\widetilde L}_{\rm f}] \in H_1(S^3 - {\widetilde L}) $ of the link 
  ${\widetilde L}_{\rm f}$, which is the union of the framings 
\be 
 {\widetilde L}_{\rm f}=\bigcup_i \widetilde C_{i {\rm f}} \; , 
 \label{25}
 \ee
 is just 
 \be 
[{\widetilde L}_{\rm f}] = \sum_{i, j} {\ell} k ( \widetilde C_{i {\rm f}} , \widetilde C_j  ) \, g_j = \sum_{i,j}   {\widetilde {\mathbb L}}_{ij} \, g_j     \; . 
  \label{26}
 \ee
By comparing equations (\ref{22}) and (\ref{26}) one finds  that the value of the abelian link invariant $\langle W( L) \rangle \bigr |_{S^3} = W( {\widetilde L}) \rangle \bigr |_{S^3}$ is completely determined by the homology class $[{\widetilde L}_{\rm f}]$ of ${\widetilde L}_{\rm f}$ in $ H_1(S^3 - {\widetilde L} )$.

The abelian link invariant (\ref{22})  also admits the following interpretation. Let $S_{\widetilde L}$ be a Seifert surface associated with the link ${\widetilde L} \subset S^3$; $S_{\widetilde L}$ is connected oriented (bicollared) with boundary $\partial S_{\widetilde L} = {\widetilde L}$. Let us denote by ${\widetilde L}_{\rm f} \cap S_{\widetilde L} $ the sum ---by taking into account the signs--- of the intersections  of the link ${\widetilde L}_{\rm f}$ with the surface $S_{\widetilde L} $. Then, as a consequence of a possible definition of the linking number \cite{ROL}, one has 
 \be 
\langle W( L) \rangle \bigr |_{S^3} = \langle W({\widetilde L}) \rangle \bigr |_{S^3} = 
\exp \Bigl \{ -( 2 i \pi / 4k)  \, {\widetilde L}_{\rm f} \cap S_{\widetilde L}   \Bigr  \} \,  . 
\label{27}
\ee

\subsection {Sum of knots and cyclic covering}  We now describe another possible interpretation of the abelian link invariants which makes use of the coverings of the complement of the links. Let us first introduce the concept of sum of knots. 

 \vskip 0.3 truecm

\begin{figure}[htbp]
\centerline{\includegraphics[width=2.70in]{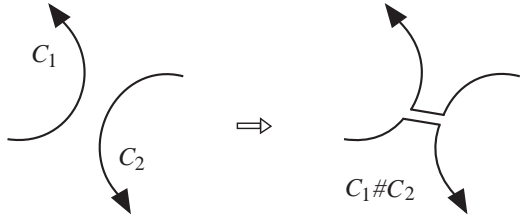}}
\caption {Sum of knots $C_1$ and $C_2$.}
\label{fig1}
\end{figure}

\vskip 0.1 truecm
 
\begin{definition}   Let $C_1$ and $C_2$ be two oriented and framed components of a link $L$, and let both components $C_1$ and $C_2$ have the same colour $q$.  By joining  $C_1$ and $C_2$ in the way shown in Figure~1, one obtains the knot $C_1 \# C_2$, that we call the {\it sum of $C_1$ and $C_2$}. The framing $(C_1 \# C_2)_{\rm f} $ of $C_1 \# C_2$ is defined to be the sum of the framings $C_{1{\rm f}} \# C_{2{\rm f}} $ so that 
\be 
\ell k   ( (C_1 \# C_2)_{\rm f} , C_1 \# C_2  ) = {\ell} k ( C_{1{\rm f}} , C_1 ) + {\ell} k (  C_{2{\rm f}} , C_2 ) + 2 \, {\ell} k ( C_1 , C_2 ) \; . 
\label{28}
\ee
\end{definition}

\begin{remark}
Note that,  when $C_1$ and $C_2$ belong to disjoint balls, the sum  $C_1 \# C_2$ coincides with the connected sum \cite{ROL} of $C_1$ and $C_2$; in general, $C_1$ and $C_2$ may be linked and tied together.  Note also that the linking number of $C_1 \# C_2$ with a generic component $C_j $ of the link $L$, with $j \geq 3 $, is just the sum of the linking numbers 
\be 
\ell k   ( C_1 \# C_2 , Cj  ) = {\ell} k ( C_1 , C_j ) + {\ell} k ( C_2 , C_j ) \; . 
  \label{29}
\ee
\end{remark}

Now, the value of each observable (\ref {12}) is invariant under the replacement of 
$C_1$ and $C_2$ by their sum $C_1 \# C_2$.  Indeed, as a consequence of the substitution of $C_1$ and $C_2$ with the sum $C_1\# C_2$, the linking matrix gets modified; instead of the first two rows and the first two columns of $   {\mathbb L}_{ij} $ one has a new single row and a new single columm.  But the relations (\ref{28}) and (\ref{29}) imply that the sum $ \sum_{ij} q_i  {\mathbb L}_{ij} q_j $ remains unchanged. 

\begin{definition}   For each coloured, oriented and framed link $L$ in $S^3$, consider its simplicial satellite $\widetilde L$. All the components of $\widetilde L$ have the same (unitary) colour;  therefore, one can recursively take the sum of the components of $\widetilde L$ so that, in the end, one obtains a single knot $L^\#$ that we call an {\it equivalent knot}  of $L$. 
\end{definition}

\noindent By construction 
\be 
 \langle W(L) \rangle \Bigr |_{S^3} = \langle W( L^\#) \rangle \Bigr |_{S^3}  \; , \; \; \forall L\subset S^3  \; . 
 \label{30}
 \ee 

Consider now an equivalent knot $L^\# $ of the link $L$ and let $L^\#_{\rm f}$ be the framing of $L^\#$. From equation (\ref{30}) it follows 
\be 
\langle W( L) \rangle \bigr |_{S^3} = \langle W(L^\# ) \rangle \bigr |_{S^3} = 
\exp \Bigl \{ -( 2 i \pi / 4k)  \, \ell k ( L^\#_{\rm f} , L^\# )   \Bigr  \} \,  . 
\label{31}
\ee
This equation shows that the expectation value $\langle W( L) \rangle \bigr |_{S^3} $ is fixed by the homology group $H_1(S^3 - L^\# )$.

Finally,   let $S_{L^\#}$ be a Seifert surface associated with  $L^\# \subset S^3$. Let $\pi :  Y_\infty  \rightarrow Y $ be the infinite cyclic cover of $Y=S^3 - {L^\#}$ and let $\tau $ be the generator of covering translations which generates Aut($ Y_\infty $).  $Y_\infty $ can be obtained \cite{ROL} by gluing ---in a orientation preserving way--- infinite copies of $S^3 - S_{L^\#}$ along (the images of) the surface $S_{L^\#}$.  Consider $L^\#_{\rm f}$ as a loop in $Y$  based at, for example,  the point $y_0 \in  L^\#_{\rm f}$.  Let the path $( L^\#_{\rm f} )_\infty $ in $Y_\infty $  be the lifting of $L^\#_{\rm f} $ based at a chosen point $(y_0)_\infty \in \pi^{-1}(y_0)$ and let $(y_1)_\infty \in \pi^{-1} (y_0) $ be the terminal point of the path $ (L^\#_{\rm f} )_\infty $.  There is a unique integer $n $ such that   $\tau^n ( y_0)_\infty =  ( y_1)_\infty$; this integer precisely determines the value of the observable 
 \be 
\langle W( L) \rangle \bigr |_{S^3} = \langle W(L^\#) \rangle \bigr |_{S^3} = 
\exp \Bigl \{ -( 2 i \pi / 4k)  \, n   \Bigr  \}    
\label{32}
\ee
because,  in agreement with equation (\ref{27}), in going along $L^\#_{\rm f}$, $n$ simply counts (by taking into account the signs) how many times one runs across the Seifert surface  $S_{L^\# }$.  Expression (\ref{32}) is periodic in $n$ with period $4k$; so, for fixed integer $k$, instead of the infinite cyclic cover $Y_\infty $, one can actually consider the $4k$-fold cyclic cover of $Y=S^3 - {L^\# }$. 

\subsection {Homology spheres}  In order to study the properties of the observables in homology spheres, we need to recall the meaning  of the surgery instruction which is described by a framed surgery link ${\mathcal L} \subset S^3$. 
Let $\{ {\mathcal L}_i \}$ (with $i=1,2,..., N_{\mathcal L}$) be the  link components of ${\mathcal L}$ with framings $\{ {\mathcal L}_{i {\rm f}} \}$. The 3-manifold $M_{\mathcal L}$  which corresponds to $\mathcal L$ can be obtained by means of the following operations; for each link component ${\mathcal L}_i $,  

\begin{itemize}

\item {}  remove from $S^3$ the interior $\! \! \vint $ of a tubular neighbourhood $V_i$ of the component ${\mathcal L}_i$; 

\item {} sew the solid torus $V_i$ on   $S^3 - \! \! \vint  \, $  by means of the boundaries identification given by a homeomorphism $h_i : \partial V_i \rightarrow \partial ( S^3 - \! \! \vint ) $ which sends a meridian $\mu_i $ of $V_i$ into the framing  $ {\mathcal L}_{i {\rm f}} $ of ${\mathcal L}_i$, i.e. $ {\mathcal L}_{i {\rm f}} = h_i ( \mu_i )\in \partial  ( S^3 - \! \! \vint  ) $. 

\end{itemize}

\noindent One example of surgery is depicted in Figure~2; in this case, the surgery link coincides with the trefoil knot with surgery coefficient 2. 

 \vskip 0.3 truecm

\begin{figure}[htbp]
\centerline{\includegraphics[width=2.50in]{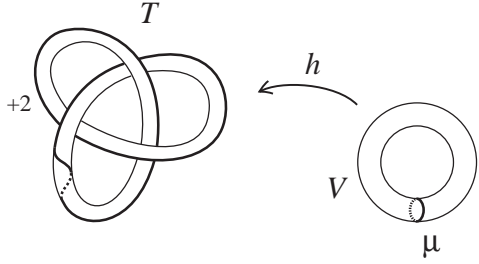}}
\caption{One example of surgery over the trefoil $T$ with surgery coefficient +2;  according to the homeomorphism $h$, the image of the meridian $\mu$ of the solid torus $V$  is a framing of $T$ with linking number $+2$ with $T$.}
\label{fig2}
\end{figure}

\vskip 0.1 truecm

Let $g_i $ be the generator of  $H_1 ( S^3  - {\mathcal L} )$ which is associated with the link component ${\mathcal L}_i$, with $i=1,2,..., N_{\mathcal L}$, where an orientation has been introduced for each  ${\mathcal L}_i$. Since the meridians $\{ \mu_i \}$ of the tubular neighbourhoods $\{  V_i \}$ are homologically trivial, their images $ {\mathcal L}_{i {\rm f}} = h_i ( \mu_i )$ also must be homologically trivial. So,  
the homology group $ H_1 ( M_{\mathcal L} ) $ of the manifold $M_{\mathcal L}$ admits the following presentation ---with $N_{\mathcal L} $ generators and at most $N_{\mathcal L} $ nontrivial relations--- 
\be
 H_1 ( M_{\mathcal L} ) = \langle \, g_1 , g_2 , \cdots  ,  g_{N_{\mathcal L}}  \; |\,  [{\mathcal L}_{1 {\rm f}}]  =0 , [ {\mathcal L}_{2 {\rm f}}] =0 , \cdots \, \rangle \; ,   
 \label{33}
 \ee
with $  [{\mathcal L}_{i {\rm f}}]  \in H_1 ( S^3  - {\mathcal L} )$ for $i=1,2,..., N_{\mathcal L}$. If $M_{\mathcal L}$ is a homology sphere, $ H_1 ( M_{\mathcal L} ) $ must be trivial. In this case, the set of relations 
\be 
[{\mathcal L}_{i {\rm f}}]  \equiv   \sum_j {\ell} k ( {\mathcal L}_{i {\rm f}} , {\mathcal L}_j  )\,  g_j = 0 \quad , \quad {\rm ~for~ } i =1,2,..., N_{\mathcal L} 
  \label{34}
 \ee
must only admit the trivial solution $g_1 = g_2 = \cdots =0$. This means that,  when ---by means of Kirby moves--- the linking matrix of the surgery link is reduced in diagonal form, each diagonal matrix element must  coincide with $+ 1$ or $-1$. In fact, the following theorem  has been proved \cite {SA}.   

\begin{theorem} \label {TH} Each homology 3-sphere admits a surgery presentation in $S^3$ described by a surgery link ${\mathcal L}$ which is algebraically split  
(i.e. ${\ell} k ( {\mathcal L}_i , {\mathcal L}_j  )=0 \; \forall \, i \not= j$)  and has surgery coefficients equal to $\pm 1$  (i.e. ${\ell} k ( {\mathcal L}_{i {\rm f}} , {\mathcal L}_i  ) = \pm 1 \; \forall \, i $). 
\end{theorem}

\noindent We can now demonstrate the following result. 

\begin{theorem}
The sets of the abelian Chern-Simons observables in $S^3$ and in any homology  3-sphere $M_0$  ($H_1(M_0)=0$) coincide,  
\be
\biggl \{ \langle W(L) \rangle \Bigr |_{M_0} \biggr \} = \biggl \{ \langle W(L) \rangle \Bigr |_{S^3} \biggr \} \; . 
\label{18}
\ee
\end{theorem}

\begin{proof}  Let $M_0$ be a homology sphere and let ${\mathcal L} \subset S^3$ be a  surgery link, which corresponds to a surgery presentation of $M_0$ in $S^3$, such that  the properties specified by  Theorem~3.4 are satisfied (that is, ${\mathcal L}$ is algebraically split with surgery coefficients $\pm 1$).     Any link $L$ in $M_0$ can be described by a link, that we shall also denote by $L$,  in the complement of ${\mathcal L}$ in $S^3$.   In order to compute the observable 
$ \langle W ( L) \rangle \bigr |_{M_0} $ we shall use the surgery method described in equation (\ref{15}). The denominator of expression (\ref{15}) contains the expectation value 
\be
\ba
 \langle  W({\mathcal L}) \rangle \Bigr |_{S^3} &=&   \sum_{q_1, q_2,...} \exp \Bigl \{ -( 2 i \pi / 4k)  \sum_{ij} q_i   \, {\ell} k ( {\mathcal L}_{i {\rm f}} , {\mathcal L}_j  )\, q_j  \Bigr  \} = \\
&=&   \prod_{i}\sum_{q_i}  \exp \Bigl \{ -( 2 i \pi / 4k) q^2_i   \, {\ell} k ( {\mathcal L}_{i {\rm f}} , {\mathcal L}_i  )   \Bigr  \} \; . 
\ea
\label{35}
\ee
Let us consider each term of the product entering equation (\ref{35}); since ${\ell} k ( {\mathcal L}_{i {\rm f}} , {\mathcal L}_i ) = \pm 1 $  one finds \cite{HL}
\be
\sum_{q=0}^{2k-1} \exp \Bigl \{ \pm ( 2 i \pi / 4k) q^2     \Bigr  \} = e^{\pm i\pi /4}  \sqrt {2k} 
\; , 
\label{36}
\ee
and then $ \langle  W({\mathcal L}) \rangle \bigr |_{S^3} \not= 0$.  This means that equation (\ref{15}) is well defined; let us now consider the numerator of equation (\ref{15}). 

Let us denote by $L_\alpha  $, with $\alpha= 1,2,..., N_{L}$,  the $\alpha$-th component of the link $L$ with colour $q_\alpha $, and  let 
\be
 t_i = \sum_\alpha q_\alpha \, \ell k ( {\mathcal L}_i , L_\alpha ) \; . 
\label{37}
\ee
Then, in the computation of the numerator   $\langle W(L) W({\mathcal L}) \rangle \bigr|_{S^3}$   of the ratio (\ref{15}),  the contribution of the generic component ${\mathcal L}_i$ of the surgery link $\mathcal L$ is given by the multiplicative factor 
\be  
\sum_{q_i=0}^{2k-1} \exp \Bigl \{ - ( 2 i \pi / 4k) \left [ (\pm   q_i^2)    + 2 q_i t_i \right ]  \Bigr  \} = 
e^{(\mp ) i\pi /4}  \sqrt {2k} \, \exp \Bigl \{ - ( 2 i \pi / 4k) ( \mp  t_i^2 ) \Bigr \} \, . 
\label{38}
\ee
In the computation of ratio (\ref{15}), the term $e^{(\mp ) i \pi /4}  \sqrt {2k}$ cancels with the same factor appearing in the denominator, see equation (\ref{36}). Whereas the remaining term $\exp \bigl \{ - ( 2 i \pi / 4k) ( \mp  t_i^2 ) \bigr \} $ corresponds to the effect of a $(\mp  1)$ twist homeomorphism of the link components of $L $ which are linked with ${\mathcal L}_i$. 

So, in the computation of $ \langle W ( L) \rangle \bigr |_{M_0} $, the global effect of the surgery link ${\mathcal L}$ is just to introduce of  certain number of twist homeomorphisms on the link $L$ whose expectation value has eventually to be computed in $S^3$.   This means that,  for each link $L\subset M_0$ one finds a suitable link $L^\prime \subset S^3$ such that $ \langle W ( L) \rangle \bigr |_{M_0} = \langle W ( L^\prime ) \rangle \bigr |_{S^3} $. Consequently, the sets of expectation values 
$\left \{ \langle W(L) \rangle \bigr |_{M_0} \right \} $ and  $ \left \{ \langle W(L) \rangle \bigr |_{S^3} \right \}$ coincide; this concludes the proof. 

\end{proof}

\begin{remark}
We would like to present now a different proof of Theorem~3.5 which is not based on algebraic manipulations. The new proof makes use of the properties of the Kirby moves and is entirely based on the fact that the abelian link invariants only depend on the linking numbers between the link components.   The starting point is that a function of the abelian link invariants, which provides a realization of the surgery rules, exists (equations (\ref{15}) and (\ref{16})).  Let us consider a surgery presentation of the homology sphere $M_0$ in $S^3$  which is described by a surgery link ${\mathcal L}$; according to Theorem \ref{TH}, one can assume that  ${\mathcal L}$ is   algebraically split and has surgery coefficients $\pm 1 $. Since all the linking numbers between the link components of $\mathcal L$ are vanishing,  the components ${\mathcal L}_i$ can be untied so that one obtains the distant union of knots, each with surgery coefficient $\pm1 $.  By means of a finite number of overcrossing/undercrossing exchanges, each knot can be unknotted. Thus, for each surgery knot  one can introduce \cite {ROL} (by means of Kirby moves)  a finite number of new elementary surgery components  ---which are given by unknots with surgery coefficients $\pm 1$---  which unknot the knot and have vanishing linking number with the knot. One example of this move is shown in Figure~3. 

\vskip 0.3 truecm

\begin{figure}[htbp]
\centerline{\includegraphics[width=1.98in]{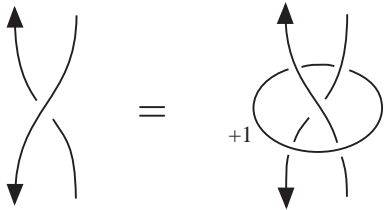}}
\caption{One example of Kirby move: by means of the introduction of the new surgery component, an undercrossing is replaced by an overcrossing.}
\label{fig3}
\end{figure}

\vskip 0.1 truecm

\noindent  Again, since all the linking numbers are vanishing,  these new link components also can be untied completely. As a result, the entire set of surgery instructions is effectively described  by the distant union of unknots with surgery coefficients $\pm 1$. The action of this surgery on $S^3$  is trivial, it maps $S^3$ into $S^3$,  because it simply introduces a set of disjoint  elementary $(\pm 1 )$ twist homemorphisms which possibly act on  the links $L\subset S^3$.  This is precisely in  agreement with the conclusions that have been obtained in the previous algebraic version of the proof. So, the set of the abelian Chern-Simons  observables in a generic homology 3-sphere $M_0$ coincides with the set of observables in the 3-sphere $S^3$. 

\end{remark}

\subsection {Homologically trivial links} 

In this section we shall consider homologically trivial links in a generic manifold $M$. Let us  firstly  introduce a general property of the expectation values. 

\begin{property}
Suppose that the link $L_U $ is the union of the link $L$ and of the unknot $U$ in a generic manifold $M$,  $L_U  = L \cup U\subset M$.   
If the unknot $U$ belongs to a 3-ball which is disjoint from the link $L $ and $U$ has trivial framing (i.e. its framing $U_{\rm f}$ satisfies $\ell k (U , U_{\rm f} ) =0$), then, independently of the colour $q$ associated with $U$, one finds 
\be
\langle W(L_U ) \rangle \Bigr |_{M} = \langle W (U) W(L) \rangle \Bigr |_{M} = \langle W(L) \rangle \Bigr |_{M} \; . 
\label{39}
\ee
\end{property}

\begin{proof}
When $M=S^3$, equation (\ref{39}) follows immediately from expression (\ref{12}). In the case of a generic 3-manifold $M$, equation (\ref{39}) is a consequence of definition (\ref{15}) and of the ambient isotopy invariance of the observables. In facts,   if $U$ belongs to a 3-ball, $U$ is ambient isotopic with an unknot  which belongs to a 3-ball which is disjoint from the surgery link ${\mathcal L}$ entering equation (\ref{15}) ---and then this unknot is  not linked with ${\mathcal L}$. Moreover, since  $L_U $ is the distant union of $U$ and $L$, the unknot $U$ is not linked with the components of the link $L$; finally, $U$ has trivial framing and then equality (\ref{39}) follows. 
\end{proof}

Let us now consider the observable $\langle W ( L) \rangle \bigr |_{M}$ which is associated with a link $L$ in a generic  3-manifold $M$. 

\begin{theorem}
Let  $M$ be a generic closed and oriented 3-manifold;  if the simplicial satellite ${\widetilde L} \subset M  $ of the link $L  $   is homologically trivial and the associated observable (\ref{15})  is well defined, then there exists a link ---that we  denote by $L^\prime$--- in $S^3$   such that   
\be
 \langle W(L) \rangle \Bigr |_{M}  =  \langle W(L^\prime) \rangle \Bigr |_{S^3}  \; . 
\label{19}
\ee
\end{theorem}

\begin{proof}
In agreement with the surgery recipe of equations (\ref{15}) and (\ref{16}),  one has 
\be
\langle W ( L) \rangle \Bigr |_{M} =   
\langle W ( \widetilde L) \rangle \Bigr |_{M} =  
 \langle W ( \widetilde L ) \, W({\mathcal L}) \rangle  \Bigr |_{S^3}  \; {\Big  / } \;   \langle  W({\mathcal L}) \rangle \Bigr |_{S^3} \, ,
\label{40}
\ee
where  $\widetilde L$ is the simplicial satellite of $L$ and $\mathcal L$ is a surgery link ---with components $\{ {\mathcal L}_i \} $--- which corresponds to $M$. We assume that expression (\ref{40}) is well defined. Let us consider the presentation (\ref{33}) of the group $H_1(M)$. If the link $\widetilde L$  is homologically trivial  in $M$,  the class $[\widetilde L ]$ of $\widetilde L $ in $H_1 ( S^3 - {\mathcal L} ) $ can be written in the form
\be
H_1 ( S^3 - {\mathcal L} ) \ni [\widetilde L ] = \sum_i n_i \, [ {\mathcal L}_{i {\rm f}}] \; , \; \; \hbox {with~}  n_i \in {\mathbb Z} \; .  
\label{41}
\ee
where $ [ {\mathcal L}_{i {\rm f}}] \in H_1(S^3 - {\mathcal L})$.  Indeed,  each component $ {\mathcal L}_{i {\rm f}} $ is homologically  trivial in $H_1 (M)$ and, in the presentation (\ref{33}), all the constraints are precisely generated by the equations $[ {\mathcal L}_{i {\rm f}}]  =0$ for $i=1,2,..., N_{\mathcal L}$.

\noindent Let us now consider the new  link $L^\prime \subset M$ which, in the surgery presentation of $M$, is described by the link $L^\prime \subset S^3 - {\mathcal L}$ given by 
\be 
 L^\prime = \widetilde L \cup K_1 \cup K_2 \cup \cdots \cup K_m  \; ,    
 \label{42}
\ee
where  each knot $K_j $, with $j=1,2,...,m$, is an unknot with unitary colour  and $ m=\sum_j n_j$. More precisely, each of the first $n_1$ unknots of equation (\ref{42}),  $K_1, K_2, ..., K_{n_1} $,  is ambient isotopic with ${\mathcal L}_{1\rm f}$ with reversed orientation; each of the next $n_2$ unknots $K_{n_1 +1}, K_{n_1 +2}, ..., K_{n_1+n_2} $, is ambient isotopic with ${\mathcal L}_{2\rm f}$ with reversed orientation and so on. If $K_j$ is ambient isotopic with ${\mathcal L}_{i\rm f}$ with reversed orientation, the framing $K_{j \rm f}$ of $K_j$ is chosen in such a way that  $\ell k( K_j , K_{j \rm f})= \ell k (  {\mathcal L}_{i}, {\mathcal L}_{i\rm f} )$. According to the surgery instructions described in section~3.3, each component ${\mathcal L}_{i\rm f}$ is homeomorphic with a meridian of a solid torus and then ${\mathcal L}_{i\rm f}$ is ambient isotopic with an unknot which belong to a 3-ball that is disjoint from all the remaining link components. Consequently, each knot $K_j$ is ambient isotopic with an unknot which belongs to a 3-ball in $M$ and, by construction, this unknot has trivial framing.  So, in agreement with the Property~3.7, one has 
\be
\langle W ( L) \rangle \Bigr |_{M} =  \langle W ( L^\prime ) \rangle \Bigr |_{M} =    
 \langle W (  L^\prime ) \, W({\mathcal L}) \rangle  \Bigr |_{S^3}  \; {\Big  / } \;   \langle  W({\mathcal L}) \rangle \Bigr |_{S^3} \, . 
\label{43}
\ee
The class $[L^\prime ]$ of $L^\prime $ in $S^3 - {\mathcal L}$ follows from the definition (\ref{42}) and equation (\ref{41}), 
\be
H_1 ( S^3 - {\mathcal L} ) \ni [L^\prime ] = \sum_i n_i \, [ {\mathcal L}_{i {\rm f}}]  + \sum_{j}  [K_j] = 0 \; . 
\label{44}
\ee
This means that $L^\prime $ ---or its equivalent knot ${L^\prime}^\#$--- is not linked with each of the components of the surgery link ${\mathcal L}$. Consequently, in the computation of the ratio (\ref{43}), the expectation value $ \langle  W({\mathcal L}) \rangle \bigr |_{S^3}$ factorizes in the numerator and cancels out with the denominator, and finally one obtains
\be
\langle W ( L) \rangle \Bigr |_{M} =    
 \langle W (  L^\prime )   \Bigr |_{S^3}  \; . 
 \label{45}
\ee
To sum up,  if the simplicial satellite  ${\widetilde  L} \subset M $ of the link $L$  is homologically trivial, there exists a link $L^\prime \subset S^3$   such that equation (\ref{45}) is satisfied, and this concludes the proof. Finally, because of the colour periodicity property of observables, for fixed integer $k$ Theorem~3.8 actually holds when $\widetilde L$ is homologically trivial mod~$2k$. 

\end{proof}

\section{Three-manifold invariants}

The surgery rules and the 3-manifold invariant of Reshetikhin and Turaev \cite{RT} for the  gauge group $SU(2)$ can be generalized to the case in which the gauge group is $U(1)$. Let $M = M_{\mathcal L}$ be the 3-manifold which is obtained by means of the Dehn surgery  which is described by the surgery link ${\mathcal L}$ in $S^3$; the 3-manifold invariant $I_k(M )$ for the abelian $U(1)$ gauge group takes the form 
\be 
I_k(M) = I_k(M_{\mathcal L} ) = \left ( 2k \right )^{- N_{\mathcal L}/2} \, e^{i\pi   \sigma ({\mathcal L}) /4} \langle W({\mathcal L} ) \rangle \Bigr |_{S^3}\; , 
\label{46}
\ee
where  $N_{\mathcal L}$ denotes the number of components of $\mathcal L$ and $\sigma ({\mathcal L})$ represents the so-called signature of the linking matrix associated with ${\mathcal L}$, i.e. $\sigma ({\mathcal L}) = n_+ - n_- $ where $n_{\pm }$ is the number of positive/negative eigenvalues of the linking matrix which is defined by the framed  link ${\mathcal L}$. 
Expression (\ref{46}) is invariant under Kirby moves \cite{RT,  MOST, L} and therefore is invariant under orientation preserving homeomorphisms of the 3-manifold $M$. Note that the orientation of $M = M_{\mathcal L}$ is induced by the orientation of $S^3$ on which the surgery acts. 
A modification of the orientation of $M$ is equivalent to the replacement of $I_k(M ) $ by its complex conjugate $\overline {I_k(M )} $.  

\begin{remark}
The invariance under Kirby moves of expression (\ref{46}) can be used to understand the consistency of the surgery rules (\ref{15}) for the observables. In facts, if one multiplies the numerator and the denominator of equation (\ref{15}) by the same factor   $\left ( 2k \right )^{- N_{\mathcal L}/2} \, e^{i\pi   \sigma ({\mathcal L}) /4}$, the numerator and the denominator are separately invariant under Kirby moves. 
\end{remark}

Let us recall that,  according to the prescription (\ref{16}), in the computation of the expectation value $\langle W({\mathcal L} ) \rangle \bigr |_{S^3}$ one must take the sum over the values $q=0,1,..., 2k-1$ of the colour which is associated with each  component of the surgery link ${\mathcal L}$. This just corresponds to the standard Reshetikhin-Turaev prescription in the  case of gauge group  U(1). Actually, for fixed integer $k$, the Reshetikhin-Turaev invariant (\ref{46}) admits the following  natural generalization. 

\begin{definition}
As we have already mentioned, for fixed integer $k$  the colour space is isomorphic with ${\mathbb Z}_{2k}$. For each subgroup ${\mathbb Z}_{p}$ of ${\mathbb Z}_{2k}$, one can introduce the 3-manifold invariant $I_{(p,k)}(M) $ defined by 
\be 
I_{(p,k)}(M) = I_{(p,k)}(M_{\mathcal L}) =  a^{- N_{\mathcal L}/2} \, e^{i \varphi   \sigma ({\mathcal L}) } \langle W({\mathcal L} ) \rangle_{(p,k)} \Bigr |_{S^3}\; , 
\label{47}
\ee
where the positive real number $a$ and the phase factor $e^{i \varphi }$ are determined by 
\be  
\sum_{b=0}^{p-1} \exp \Bigl \{  - (  i \pi \,d  / p)  \,  b^2      \Bigr  \} = 
a\, e^{- i\varphi }   \, , 
\label{48}
\ee
in which $d = 2k / p$. $ \langle W({\mathcal L} ) \rangle_{(p,k)} \bigr |_{S^3}$ denotes the sum of the  expectation values  when the  colour of each link component takes the values $q=0, d, 2d,..., (p-1)d$.  One has $ I_{(2k,k)}(M_{\mathcal L})= I_k(M_{\mathcal L})$. 
\end{definition}

\noindent The proof that $I_{(p,k)}(M_{\mathcal L})$ is invariant under Kirby moves is based precisely on the same steps that enter the corresponding proof for  $I_k(M_{\mathcal L})$. 

A field theory interpretation of the Reshetikhin-Turaev invariant (\ref{46})  ---as a ratio of Chern-Simons partition functions--- has been proposed in \cite{L} and detailed  discussions on the properties of the invariant (\ref{46}) can be found in \cite{JP, FDE, HST}. 

\begin{property}
If the manifold $M_0$ is a homology 3-sphere, then $I_k(M_0) = 1$. 
\end{property}
\begin{proof}
By Theorem~3.4, $M_0$ admits a surgery presentation in $S^3$ in which the surgery link $\mathcal L$  is algebraically split with surgery coefficients $\pm 1$. Since the link components of $\mathcal L$ are not linked, in agreement with equation (\ref{35}) the expectation value $\langle W ({\mathcal L}) \rangle \bigr |_{S^3}$  is just the product of terms shown in equation (\ref{36}). These terms cancel with the normalization factor  $\left ( 2k \right )^{- N_{\mathcal L}/2} \, e^{i\pi   \sigma ({\mathcal L}) /4}$ which is present in the definition of $I_k(M_0)$, and then $I_k(M_0) = 1$. 
\end{proof} 

Let us consider the lens spaces $L_{p/r} $ where the  integers $p$ and $r$ are coprime and satisfy $0 < r < p$.  The fundamental group of $L_{p/r}$ is ${\mathbb Z}_p $ and one also has $H_1 (L_{p/r}) \simeq {\mathbb Z}_p$. When $p \not= p^\prime$, the lens spaces $L_{p/r}$ and $L_{p^\prime /r^\prime }$ are not homeomorphic. 
The manifolds $L_{p/r}$ and $L_{p/r^\prime }$  are homeomorphic iff $\pm r^\prime \equiv r^{\pm 1} $ (mod~$p$).   
The  manifold $L_p  $ admit a surgery presentation given by the unknot  with surgery coefficient equal to the integer $p$.
Special cases are $L_0  \simeq  S^2\times S^1$, $L_1 \simeq S^3$; equation  (\ref{46})  gives 
\be 
I_k (S^3) = 1 \quad , \quad I_k( S^2 \times S^1 ) = \sqrt {2k }  \; . 
\label{49}
\ee
By using the following reciprocity formula \cite{SIE}
\be 
\sum_{n=0}^{|c|-1} e^{-{i\pi \over c} ( a n^2 +bn) } = \sqrt{| {c / a}| } \, e^{- {i \pi \over 4 ac } ( |ac | - b^2) }\, 
\sum_{n=0}^{|a|-1} e^{{i\pi \over a} ( c n^2 +bn) } \; , 
\label{50}
\ee
which is valid for integers $a$, $b$ and $c$ such that $ac \not= 0$ and $ac + b = \, $even,   one gets (for integer  $  p >1$)
\be 
I_k(L_p) = \frac{1}{\sqrt {p}}  \sum_{n=0}^{p-1} e^{ 2\pi i k \, ( n^2 / p)} \; . 
\label{51}
\ee
Let us compare expression (\ref{51})  with the functional integral interpretation  \cite{L} of the Reshetikhin-Turaev invariant
\be 
I_k(L_p) = \frac{\int_{L_p} DA \, e^{2 \pi i k S }  }{\int_{S^3} DA \, e^{2 \pi i k S }} \; .   
\label{52}
\ee
In agreement with equation (\ref{8}),  the fields configuration space $H^1_D (L_p)$ has the structure of a bundle over $H^2(L_p) = {\mathbb Z}_p $ with fibre $\Omega^1 (L_p) \big / \Omega_{\mathbb Z}^1(L_p)$. Let the group $H^2(L_p)$ be generated by the element $h$, with $h^p =1$. 
The path-integral (\ref{52})  over $H^1_D (L_p)$ is formally given by a sum of $p$ terms; the $n$-th  term corresponds to the path-integral over 1-forms modulo forms of integer periods in the $n$-th sector of $H^1_D (L_p)$ which is associated to the  element $h^n$  of the second  cohomology group of $L_p$. The result (\ref{51})  suggests the possibility that the  path-integral in the $n$-th sector of $H^1_D(L_p)$ is saturated by a single  value $S \big |_{n}  $ of the Chern-Simons action, with $S \big |_{n} = n^2 /p $ modulo  integers. 

The manifold $\Sigma_g \times S^1$, where $\Sigma_g$ denotes a closed oriented surface of genus $g$, admits a surgery presentation that is described \cite{EGA} by a surgery link which contains $2g+1$ components (with vanishing surgery coefficients) and has vanishing linking matrix. The first homology group is $H_1(\Sigma_g \times S^1 ) = {\mathbb Z}^{2g+1} $ ; one finds 
\be
I_k (\Sigma_g \times S^1)= (2k)^{(2g+1)/2} \; . 
\label{53}
\ee

Since the abelian link invariants only  depend on the homology of the complement of the (simplicial satellites of) links in $S^3$, one could suspect  that  the invariant $I_k(M)$  only depends on the homology group $H_1 (M)$ of the closed oriented 3-manifold $M$. 
This guess is  supported by Property~4.3 and by the result (\ref{53}); moreover, it naturally fits the structure of the  configuration space on which the functional integral is based. 
However, a few counter-examples  demonstrate that this conjecture  is false.  In the non-abelian case, this guess is not correct; in facts, explicit counter-examples for the gauge groups $SU(2)$ and $SU(3)$ can be found in Ref.\cite{GP}.  
   
The lens space $L_{5/1}$ admits a surgery presentation in $S^3$ which is described by the unknot with surgery coefficient 5; whereas a surgery link corresponding to  $L_{5/2}$ is the Hopf link \cite{ROL} in $S^3$  with surgery coefficients 2 and 3. From equation(\ref{46}) we obtain 
\be
I_2(L_{5/1}) = -1 \qquad , \qquad I_2(L_{5/2}) = 1 \; . 
\label{54}
\ee

\noindent The manifold $L_{9/1}$ can be described by the unknot in $S^3$ with surgery coefficient 9 and $L_{9/2}$ corresponds to the Hopf link with surgery coefficients 5 and 2.   We get 
\be
I_3(L_{9/1}) = i \sqrt{3} \qquad , \qquad I_3(L_{9/2}) = -i \sqrt{3} \; . 
\label{55}
\ee

\subsection*{Homotopy type} The manifolds  $ L_{9/1}$ and  $L_{9/2} $   are of the same homotopy type \cite{ROL}. Therefore, equation (\ref{55}) also shows that the Reshetikhin-Turaev invariant (\ref{46}) is not a   function of the homotopy type of the manifold $M$ only. In facts,  Murakami, Ohtsuki and Okada have shown \cite{JP} that expression (\ref{46}) is invariant under orientation-preserving homotopies  \cite{CBT}. Since under a modification of the orientation of the manifold $M$ one gets  $I_k(M) \rightarrow {\overline {I_k (M)}}$, the result (\ref{55}) is in agreement with Murakami, Ohtsuki and Okada statement. 

Let us consider the lens spaces with the orientation induced by the surgery presentation;  $L_{p/r}$ and $L_{p/ r^\prime }$ are of the same homotopy type iff $\pm r r^\prime \equiv  m^2 $ (mod $p$) for some integer $m$. Hansen, Slingerland and Turner have shown \cite{HST} that, when $r r^\prime \equiv - m^2$  (mod $p$),   one finds  $I_k (L_{p/r}) = {\overline {I_k(L_{p / r^\prime})}}$; one example of this kind  is shown in equation (\ref{55}).  Whereas, when  the product $r r^\prime $ is equivalent to a quadratic residue,  $r r^\prime \equiv  m^2$  (mod $p$),   one has $I_k (L_{p/r}) = I_k(L_{p / r^\prime})$, for istance
\be 
I_3(L_{7/1}) = - i  \qquad , \qquad I_3(L_{7/2}) = -i  \; . 
\label{56}
\ee

\noindent The equivalence relation under orientation-preserving homotopy extends to the manifolds which are  connected sum  of equivalent spaces \cite{CBT}. However, in the presence of orientation-reversing homotopy, this   equivalence relation in general does not survive  the connected sum.   
For instance,  the connected sums $L_{9/1
} \# L_{7/1}$ and $L_{9/2} \# L_{7/2}$ have different Reshetikhin-Turaev invariants  which are not related by complex conjugation
\be 
I_3(L_{9/1} \# L_{7/1}) = \sqrt{3}  \qquad , \qquad I_3(L_{9/2} \# L_{7/2}) = - \sqrt{3}  \; . 
\label{57}
\ee

\vskip 1 truecm 

\noindent {\bf Acknowledgments.}  We wish to thank R.~Benedetti and A.~Lerario for useful discussions. 

\vskip 1 truecm

\bibliographystyle{amsalpha}

\end{document}